\documentclass[a4paper,USenglish,cleveref,thm-restate]{lipics-template/lipics-v2021}

\PassOptionsToPackage{hyphens}{url}
\usepackage{hyperref}

\usepackage{booktabs}
\usepackage[table]{xcolor}
\usepackage{siunitx}
\usepackage{tikz}
\usepackage{xcolor}
\usetikzlibrary{fit, calc, decorations.pathreplacing}

\usepackage{caption}
\usepackage{subcaption}

\usepackage{amsmath}
\usepackage{amsfonts}
\usepackage{amssymb}

\usepackage[capitalise]{cleveref}
\usepackage{etoolbox}
\usepackage{thmtools}

\usepackage{xspace}

\usepackage[algo2e, linesnumbered, vlined, ruled]{algorithm2e}

\graphicspath{{./figures/}}

\def\ie{i.e.\xspace}
\def\eg{e.g.\xspace}

\def\Oh#1{\ensuremath{\mathcal O\!\left(#1\right)}}
\def\scan#1{\ensuremath{\text{scan}(#1)}}
\def\sort#1{\ensuremath{\text{sort}(#1)}}
\def\nproc{\mathcal{P}}

\newcommand{\nlpagen}{\normalfont{SeqPolyPA}\xspace}
\newcommand{\pnlpagen}{\normalfont{ParPolyPA}\xspace}
\newcommand{\dyngen}{\normalfont{MVNPolyPa}\xspace}

\newcommand{\dyngenrem}{\normalfont{MVNPolyPa}\textsuperscript{remove}\xspace}

\newcommand{\emgenpa}{\normalfont{EM-GenPA}\xspace}
\renewcommand{\deg}{\ensuremath{\mathrm{d}}}

\newcommand{\dv}[1]{\ensuremath{d(#1)}}

\newcommand{\dvi}[2]{\ensuremath{d_{#2}(#1)}}

\declaretheorem[style=mythmstyle,name=Definition]{mydef}

\declaretheorem[style=mythmstyle,name=Observation]{myobs}

\hideLIPIcs

\title{Parallel and I/O-Efficient Algorithms for Non-Linear Preferential Attachment} 

\titlerunning{Algorithms for Non-Linear Preferential Attachment} 

\author{Daniel Allendorf}{Goethe University Frankfurt, Germany}{dallendorf@ae.cs.uni-frankfurt.de}{}{}
\author{Ulrich Meyer}{Goethe University Frankfurt, Germany}{umeyer@ae.cs.uni-frankfurt.de}{}{}
\author{Manuel Penschuck}{Goethe University Frankfurt, Germany}{mpenschuck@ae.cs.uni-frankfurt.de}{}{}
\author{Hung Tran}{Goethe University Frankfurt, Germany}{htran@ae.cs.uni-frankfurt.de}{}{}

\authorrunning{D.~Allendorf, U.~Meyer, M.~Penschuck and H.~Tran}

\Copyright{Daniel Allendorf, Ulrich Meyer, Manuel Penschuck and Hung Tran}



\supplement{
The internal memory algorithms are maintained at \url{https://github.com/massive-graphs/nonlinear-preferential-attachment}.
The implementation of the dynamic weighted sampling data structure of~\cite{DBLP:journals/mst/MatiasVN03} is independently maintained as the Rust crate \url{https://crates.io/crates/dynamic-weighted-index}.
The external memory algorithms are available at \url{https://github.com/massive-graphs/extmem-nlpa}.
An archive containing the frozen source code of all experiments and most of the raw data collected can be found on \url{https://zenodo.org/record/7318118}.
}

\funding{This work was supported by the Deutsche Forschungsgemeinschaft (DFG) under grant ME 2088/5-1 (FOR~2975 --- Algorithms, Dynamics, and Information Flow in Networks}

\begin{document}

\maketitle

\normalsize

\begin{abstract}
Preferential attachment lies at the heart of many network models aiming to replicate features of real world networks.
To simulate the attachment process, conduct statistical tests, or obtain input data for benchmarks, efficient algorithms are required that are capable of generating large graphs according to these models.

Existing graph generators are optimized for the most simple model, where new nodes that arrive in the network are connected to earlier nodes with a probability $P(h) \propto d$ that depends linearly on the degree $d$ of the earlier node $h$.
Yet, some networks are better explained by a more general attachment probability $P(h) \propto f(d)$ for some function $f \colon \mathbb N~\to~\mathbb R$.
Here, the polynomial case $f(d) = d^\alpha$ where $\alpha \in \mathbb R_{>0}$ is of particular interest.

In this paper, we present efficient algorithms that generate graphs according to the more general models.
We first design a simple yet optimal sequential algorithm for the polynomial model.
We then parallelize the algorithm by identifying batches of independent samples and obtain a near-optimal speedup when adding many nodes.
In addition, we present an I/O-efficient algorithm that can even be used for the fully general model.
To showcase the efficiency and scalability of our algorithms, we conduct an experimental study and compare their performance to existing solutions.
\end{abstract}

\section{Introduction}
\label{sec:introduction}
Networks govern almost every aspect of our modern life ranging from natural processes in our ecosystem, over urban infrastructure, to communication systems.
As such, network models have been studied by many scientific communities and many resort to random graphs as the mathematical tool to capture them~\cite{bollobas1985random,albertbarabasi02,barabasi2016network}.

A well-known property of observed networks from various domains is scale-freeness, typically associated with a powerlaw degree distribution.
Barab\'{a}si and Albert~\cite{BarabasiAlbert99} proposed a simple random graph model to show that growth and a sampling bias known as \emph{preferential attachment} produce such degree distributions.

In fact, a plethora of random graph models rely on preferential attachment (e.g., \cite{DBLP:journals/jasis/Price76,DBLP:conf/soda/BollobasBCR03,krapivsky2001degree,DBLP:conf/cocoon/KleinbergKRRT99,PhysRevE.65.026107,cond-mat/0106144,DBLP:conf/websci/KunegisBM13}).
At heart, they grow a graph instance, by iteratively introducing new nodes which are connected to existing nodes, so-called \emph{hosts}.
Then, preferential attachment describes a positive feedback loop on, for instance, the host's degree;
nodes with higher degrees are favored as hosts and their selection, in turn, increases their odds of being sampled again.

Most aforementioned models use \emph{linear} preferential attachment and, consequently, sampling algorithms are optimized for this special case~\cite{batagelj2005efficient, DBLP:conf/alenex/MeyerP16, DBLP:journals/ipl/Sanders016, DBLP:conf/sc/AlamKM13, DBLP:conf/sac/AzadbakhtBBA16, DBLP:journals/dase/AlamPS19} (see also \cite{DBLP:journals/corr/abs-2003-00736} for a recent survey).
Some observed networks however, can be better explained with \emph{polynomial} preferential attachment~\cite{DBLP:conf/websci/KunegisBM13}.

In this paper, we focus on efficient algorithmic techniques of sampling hosts for non-linear preferential attachment.
Hence, we assume the following random graph model but note that it is straight-forward to adopt our algorithms to most established variants\footnote{In particular, we forbid loops and multi-edges, e.g. our algorithms generate \emph{simple} graphs.}.

\begin{mydef}[Preferential attachment]\label{def:pa}
We start with an arbitrary so-called \emph{seed graph} $G_0$ with $n_0$ nodes and $m_0$ edges.
We then iteratively add $N$ new nodes $v_{n_0 + 1}, \ldots, v_{n_0 + N}$ and connect each to $\ell \le n_0$ \emph{different} hosts.
The resulting graph $G_N$ has $n= n_0 + N$ nodes and $m=m_0 + N\ell$ edges.

The probability to select a node~$h$ with degree~$d$ as host is governed by $P(h) \propto f(d)$.
In the case where $f(d) = d^{\alpha}$ for some constant $\alpha \in \mathbb R_{>0}$, we speak of \emph{polynomial} preferential attachment, and if $\alpha = 1$, of \emph{linear} preferential attachment.
\end{mydef}

\subsection*{Related Work}

A wide variety of algorithms have been proposed for the linear case.
Brandes and Batagelj gave the first algorithm to run in time linear in the size of the generated graph \cite{batagelj2005efficient}.
Their algorithm exploits a special structure of linear preferential attachment: given a list of the edges in the graph, we may simply select an edge $e~= \{u, v\}$ uniformly at random and then toss a fair coin to decide on either $u$ or $v$ as host. 
\cite{DBLP:journals/ipl/Sanders016} gave a communication-free distributed parallelization of the algorithm, and \cite{DBLP:conf/alenex/MeyerP16} proposed a parallel I/O-efficient variant for use in external memory.
A different shared memory parallel algorithm was also proposed by~\cite{DBLP:conf/sac/AzadbakhtBBA16}, and~\cite{DBLP:conf/sc/AlamKM13, DBLP:journals/dase/AlamPS19} gave algorithms for distributed memory.

Note that we may simulate preferential attachment by using a dynamic data structure that allows fast sampling from a discrete distribution.
While rather complex solutions for this more general problem exist (e.g. see \cite{DBLP:journals/tomacs/RajasekaranR93, DBLP:conf/icalp/HagerupMM93, DBLP:journals/mst/MatiasVN03}), none provide guarantees on their performance in external memory, and we are not aware of any parallelization that is practical for our use case.

\subsection*{Our contribution}

We first present a sequential algorithm for internal memory (see \cref{sec:seq-nonlinear-pa}).
Our algorithm runs in expected time linear in the size of the generated graph and can be used for any case $P(h) \propto d^{\alpha}$ where $\alpha \in \mathbb R_{\geq 0}$.
It builds on a simple data structure for drawing samples from a dynamic distribution that is especially suited to the structure of updates in preferential attachment but may be of independent interest.

In \cref{sec:par-nonlinear-pa}, we show how to identify batches of independent samples for the polynomial model.
As independent samples can be drawn concurrently, we may use this technique to parallelize a sequential algorithm for the model.
Consequently, we apply this technique to our sequential algorithm and obtain a shared memory parallel algorithm with an expected runtime of $ \mathcal O((\sqrt{N} + N / \nproc) \log \nproc)$.

For the external memory setting (\cref{sec:external-generalized-pa}), we extend the algorithm of \cite{DBLP:conf/alenex/MeyerP16} by using a two-step process: in the first phase, we sample the degrees of the hosts of each node, and in the second phase, we sample the hosts of the nodes.
Our algorithm requires $\Oh{\sort{n_0 + m}}$ I/Os and even applies to the most general case where $P(h) \propto f(d)$ for any $f \colon \mathbb N \to \mathbb R$.

In an empirical study (\cref{sec:experiments}), we demonstrate the efficiency and scalability of our algorithms.
We find that our sequential algorithms incur little slowdown over existing solutions for the easier linear case.
In addition, our parallel algorithm obtains speed-ups over the sequential algorithm of $32$ to $46$ using $63$ processors.

\subsection*{Preliminaries and notation}
\label{sec:preliminaries}
Given a graph $G=(V,E)$ and a node $v \in V$, define the \emph{degree} $\dvi{v}{G} = \left|\{ u\, \colon \{u,v\} \in E \}\right|$ as the number of edges incident to node $v$, and let $\Delta_G = \max_{v \in V} \dvi{v}{G}$ denote the maximum degree in $G$.
For a sequence of graphs $G_1, \dots, G_N$, we also write $\dvi{v}{G_i}$ as $\dvi{v}{i}$ and $\Delta_{G_i}$ as $\Delta_i$ or drop the subscript if clear by context.

We analyze our parallel algorithm using the CREW-PRAM model (see \cite{jaja1992introduction}) on $\nproc$ processors (PUs).
This machine model allows concurrent reads of the same memory address in constant time, but disallows parallel writes to this same address.
Potentially conflicting writes to the same addresses can, however, be simulated in time $\Oh{\log \nproc}$ per access~\cite{jaja1992introduction}.
We use two methods of conflict resolution, namely minimum (storing the smallest value written to an address) and summation (storing the sum of values);
both are practical, since they are either directly supported by modern parallel computers or can be efficiently simulated.

For the analysis of our external memory algorithms we use the commonly accepted I/O-model of Aggarwal and Vitter~\cite{DBLP:journals/cacm/AggarwalV88}.
The model features a memory hierarchy consisting of two layers, namely the fast internal memory holding up to $M$ items, and a slow disk of unbounded size.
Data between the layers is transferred using so-called I/Os, where each I/O transfers a block of $B$ consecutive items.
Performance is measured in the number of I/Os it requires.
Common tasks of many algorithms include: (i) scanning $n$ consecutive items requires $\scan{n} := \Theta(n/B)$ I/Os, (ii) sorting $n$ consecutive items requires $\sort{n} := \Theta \big((n/B)\cdot\log_{M/B}(n/B)\big)$ I/Os, and (iii) pushing and popping $n$ items into an external priority queue requires $\Oh{\sort{n}}$ I/Os \cite{DBLP:journals/algorithmica/Arge03}.

\section{Sequential Algorithm}
\label{sec:seq-nonlinear-pa}
In this section, we describe a sequential internal memory algorithm for polynomial preferential attachment, \ie the probability to select node~$h$ as host is $P(h)~\propto~\dv{h}^\alpha$ for a constant $\alpha \in \mathbb R_{\geq 0}$.
A parallelization is discussed in \cref{sec:par-nonlinear-pa}.
To simplify the description, we assume that new nodes only select one host ($\ell = 1$); generalization is straight-forward by drawing multiple hosts per node and rejecting duplicates.

We also note that the algorithm can in principle be used for any function $f(d)$ that is non-decreasing.
However, analyzing the runtime for other $f(d)$ is not trivial.
In particular, \cref{th:seq-linear-space} requires the pre-asymptotic degree distribution of the generated graphs to have certain properties, but the exact distribution is not known even in the polynomial case.

\subsection{Sampling Method}
\label{subsec:seq-ars}

To simulate preferential attachment, we require an efficient sampling method which selects a host in each step.
Depending on the seed graph and parameters chosen, the host distribution $P(h)$ may also undergo significant changes as more and more nodes are added to the graph\footnote{We refer the interested reader to \cite{krapivsky2000connectivity} for an analysis of the degree distributions of non-linear preferential attachment graphs.}.
Therefore, the method used should be capable of adapting to these changes in order to be efficient for adding any number of nodes to a graph.

Our sampling method builds on rejection sampling, a general technique for sampling from a target distribution by using an easier proposal distribution.
To showcase rejection sampling, we consider a simple but suboptimal scheme for polynomial preferential attachment.
Let the graph be given as an array of edges $E = [e_1, \dots, e_m]$ in arbitrary order.
Now, sample an edge $e_i = \{u, v\}$ uniformly at random, then, randomly choose either $u$ or $v$ using a fair coin.
Based on the observation that node~$h$ with degree~$\dv{h}$ appears $\dv{h}$ times in $E$, we propose $h$ with probability $\dv{h} / 2m$.
Now, if $\alpha < 1$, accept $h$ with probability $(1 / \dv{h})^{1 - \alpha}$, or, if $\alpha \geq 1$, accept $h$ with probability $(\dv{h} / \Delta)^{\alpha - 1}$, otherwise, restart with a new proposal.

Observe that this scheme implements rejection sampling with the host distribution of linear preferential attachment as proposal distribution.
It can be shown that this scheme provides samples in constant time if $n_0 = \Oh{1}$ as $n \to \infty$ by using known properties of the asymptotic degree distributions for $\alpha < 1$ and $\alpha \geq 1$.
However, the scheme is inefficient if we wish to add only a few nodes to a larger seed graph.
For instance, consider running the scheme for $\alpha = 2$ on a seed graph of one node with degree $\sqrt{n_0} - 1$ and $n_0 - 1$ nodes with degree $1$.
It is straightforward to check that this results in an initial acceptance probability of $\Oh{1 / \sqrt{n_0}}$.
The issue is that the seed graph degree distribution differs significantly from the limit degree distribution, which causes the acceptance probability to be small until we have added many times over the initial number of nodes.

To remedy this issue, we combine rejection sampling with a dynamic proposal distribution that adapts to the target distribution.
To this end, let $S = \{1, \dots, n\}$ be a set of indices we wish to sample from and let $f\colon S \to \mathbb R_{\geq 0}$ be a function giving the weight of each element, i.e. the distribution on $S$ is $f(i) / W$ where $W = \sum_{j \in S} f(j)$ gives the proper normalization.
The initial construction of our data structure is comparable to a variant of the alias method \cite{DBLP:conf/stoc/BringmannL13}.
We first initialize an empty list or array\footnote{
	An efficient implementation requires fast sampling from $P$ and inserting into $P$;
	an array with table doubling (\eg \cite{DBLP:books/daglib/0023376}) is an easy choice with constant expected/amortized time per operation.
} $P$, and then for each element $i \in S$, add $i$ exactly $c(i) = \lceil f(i) n / W \rceil$ times to $P$.
We call $c(i)$ the \emph{count} of $i$.
As $P$ is used to maintain the proposal distribution, we refer to $P$ as the \emph{proposal list}\footnote{
	We may also think of $P$ as a compression of the target distribution. Rejection sampling then serves the purpose of correcting any errors caused by the loss of information.
}.
Having constructed $P$, it is easy to verify that we may sample an element according to the target distribution by first selecting a uniform random element $i$ from $P$, and then accepting with probability proportional to $f(i) / c(i)$.
If the distribution changes, we update $P$ as follows: to add a new element $i$, we calculate $c(i)$ as during the construction phase and add $i$ exactly $c(i)$ times to $P$.
If the weight of an existing element $i$ increases, we recalculate its count $c(i)$ and add $i$ to $P$ accordingly.

\subsection{SeqPolyPA}
\label{subsec:seq-nlpagen}

\begin{algorithm2e}[t]
	\KwData{Seed graph $G_0 = (V_0, E_0)$, with $n_0 = |V_0|$, $m_0 = |E_0|$, requested nodes $N$, parameter $\alpha > 0$}
	$P_0 \gets [v_1, \dots, \underbrace{v_i, \dots, v_i}_{\lceil \dvi{v_i}{0}^\alpha n_0 / W_0 \rceil \text{ times}}, \dots, v_{n_0}]$\;
	\For(\hfill \CommentSty{// Add new node}){$i \in [1, \ldots, N]$}{
		\Repeat(\hfill \CommentSty{// Sample a host}){accepted}{
			Select candidate $h$ uniformly from $P_{i-1}$\;
			Accept with probability $w_{i-1}(h) / \max_{v \in V_{i-1}} w_{i-1}(v)$ where $w_{i-1}(v) := \dvi{v}{i-1}^\alpha / c_{i-1}(v)$ and $c_{i-1}(v) := \#$occurences of $v$ in $P_{i-1}$\;\label{alg:suplingen-reject}
		}
		$V_i \gets V_{i-1} \cup \{v_{n_0 + i}\}$ 		\tcp*{Update graph}
		$E_i \gets E_{i-1} \cup \{\{v_{n_0 + i}, h\}\}$\;
		$P_i \gets P_{i-1} + [v_{n_0 + i}]$ \tcp*{Update $P$}
		\While{$w_i(h) > W_i / n_i$}{
			$P_i \gets P_i + [h]$\;
		}
	}

	Return $G_N =(V_N, E_N)$\;
	\caption{\nlpagen}
	\label{algo:nlpagen}
\end{algorithm2e}

Our \nlpagen algorithm implements \cref{subsec:seq-ars} as detailed in Algorithm \ref{algo:nlpagen}.
First initialize the proposal list $P_0$ by setting the count of each node $v$ of the seed graph to the optimal count $c(v) = \lceil \dvi{v}{0}^\alpha n_0 / W_0 \rceil$.
Now, in each step $i$, a new node $v_{n_0 + i}$ arrives and has to be connected to the graph by linking it to some earlier node $h$.
To this end, select a candidate $h$ uniformly at random from $P_{i-1}$.
Then, accept $h$ with probability $w_{i-1}(h) /  \max_{v \in V_{i-1}} w_{i-1}(v)$ where $w_{i-1}(v) = \dvi{v}{i-1}^\alpha / c_{i-1}(v)$, otherwise, reject $h$ and restart with a new proposal.
Once a host $h$ has been accepted, add the new node $v_{n_0 + i}$ to $V_{i-1}$, and add the edge $\{v_{n_0 + i}, h \}$ to $E_{i-1}$ to obtain the new graph $G_i = (V_i, E_i)$.
To reflect the new distribution induced by $G_i$, adjust $P$ as follows: first, add $v_{n_0 + i}$ to $P_{i-1}$ to obtain $P_i$, then, add $h$ to $P_i$ until $w_i(h) \leq W_i / n_i$.

In the following, we establish that \nlpagen produces the correct output distribution.

\begin{theorem}
	\label{th:seq-super-cor}
	\nlpagen samples host~$h \in V$ with probability $\dv{h}^\alpha / W$.
\end{theorem}

\begin{proof}
	The probability $p(h)$ that node $h$ is accepted after a single proposal is
	\begin{align}
		\label{eq:success-h-at-first}
		\nonumber
		p(h) = \underbrace{\frac{c(h)}{|P|}}_{\text{propose $h$}} \underbrace{\frac{w(h)}{\max_{v \in V} w(v)}}_\text{accept $h$} = \frac{\dv{h}^{\alpha}}{|P| \max_{v \in V} w(v)}.
	\end{align}
	Let $W' := |P| \max_{v \in V} w(v)$.
	Then, we have
	\begin{equation}
		\nonumber
		W' = |P| \max_{v \in V} w(v) \geq \sum_{v \in V} c(v) w(v) = \sum_{v \in V} \dv{v}^{\alpha} = W.
	\end{equation}
	With the remaining probability $q = 1 - W / W'$, the first proposal is rejected, and another node is proposed.
	Thus, the overall probability of sampling node $h$ is
	\begin{equation}
		\nonumber
		p(h) + p(h) q + \dots = p(h) \sum_{k = 0}^\infty q^k = p(h) \frac{1}{1 - q} = \frac{\dv{h}^{\alpha}}{W}
	\end{equation}
	as claimed.
\end{proof}

\noindent
Next, we show that \nlpagen runs in expected time linear in the size of the generated graph.
We first analyze the memory usage.

\begin{lemma}
	\label{th:seq-linear-space}
	\normalfont Given a seed graph $G_0=(V_0, E_0)$ with $n_0 = |V_0|$ nodes, \nlpagen adds $N$ new nodes to $G_0$ using a proposal list $P$ of expected size $\Oh{n_0 + N}$ if $\alpha \leq 1$ or $\alpha > 1$ and $n_0 = \Oh{1}$\footnote{We remark that this is a correction over the conference version of the article \cite{allendorf2023parallel} in which the condition $n_0 = \Oh{1}$ if $\alpha > 1$ was missing.}.
\end{lemma}

\begin{proof}
	The initial size of $P$ is at most
	\begin{equation}
		\nonumber
	 	|P_0| = \sum_{v \in V_0} \left\lceil \frac{\dvi{v}{0}^\alpha n_0}{W_0} \right\rceil \leq \sum_{v \in V_0} \left(1 + \frac{\dvi{v}{0}^\alpha n_0}{W_0} \right) = 2 n_0.
	\end{equation}
	Then, in each step $i$, the size of $P$ increases by $1$ for the new node and by the increase in the count of the host node $h$.
	The increase in the count is
	\begin{equation}
		\nonumber
		\max \left\{0, \left\lceil \frac{\dvi{h}{i}^\alpha n_i}{W_i} \right\rceil - \left\lceil \frac{\dvi{h}{j}^\alpha n_{j}}{W_{j}} \right\rceil \right \}
	\end{equation}
	where $j < i$ is the last step in which node $h$ was sampled.
	
	Now, we distinguish two cases.
	For $\alpha \leq 1$, we have $\dvi{h}{i}^\alpha = (\dvi{h}{j} + 1)^\alpha \leq \dvi{h}{j}^\alpha + 1$, and using $n_i / W_i \leq 1$ and $W_i > W_j$, we obtain
	\begin{align}
		\nonumber
		\left\lceil \frac{\dvi{h}{i}^\alpha n_i}{W_i} \right\rceil - \left\lceil \frac{\dvi{h}{j}^\alpha n_{j}}{W_{j}} \right\rceil < 2 + \frac{\dvi{h}{j}^\alpha}{W_i} \left(n_i - n_j\right).
	\end{align}
	Now, observe that $\dvi{h}{j}^\alpha / W_i$ is the probability that $h$ is sampled in step $i$, and as $W$ is non-decreasing, the probability that $h$ is sampled in each step between $j$ and $i$ is at least $\dvi{h}{j}^\alpha / W_i$.
	Thus, the expected length of a run until $h$ is sampled is at most $W_i / \dvi{h}{j}^\alpha$, and as by definition, $i - j$ is the length of this run, we have $E[i - j] = E[(n_0 + i) - (n_0 + j)] = E[n_i - n_j] \leq W_i / \dvi{h}{j}^\alpha$.
	Therefore, the expected increase is only constant, which shows the claim.
	
	The other case is $\alpha > 1$ and $n_0 = \Oh{1}$.
	In this case as a single node emerges which obtains almost all links and has expected degree $\Delta = \Theta(n)$ as $n \to \infty$ \cite{krapivsky2000connectivity}.
	In each further step $i$ we then have $W_i / n_i = \Omega(\Delta^\alpha / n_i) = \Omega(n_i^{\alpha-1})$ and the increase in the count of the chosen host node $h$ is
	\begin{align}
		\nonumber
		\left\lceil \frac{\dvi{h}{i}^\alpha n_i}{W_i} \right\rceil - \left\lceil \frac{\dvi{h}{j}^\alpha n_{j}}{W_{j}} \right\rceil &= \Oh{\frac{\dvi{h}{i}^\alpha - \dvi{h}{j}^\alpha}{n_i^{\alpha-1}}}.
	\end{align}
	In addition, we have $\dvi{h}{i} = \dvi{h}{j} + 1 < n_i + 1$, and it is straightforward to check that $(x + 1)^\alpha = x^\alpha + \Oh{x^{\alpha - 1}}$ for any $x > \alpha > 1$, so it follows that
	\begin{equation}
		\nonumber
		\frac{\dvi{h}{i}^\alpha - \dvi{h}{j}^\alpha}{n_i^{\alpha-1}} < \frac{(n_i + 1)^\alpha - n_i^\alpha}{n_i^{\alpha-1}} = \Oh{1}.
	\end{equation}
	This implies that in each further step, the increase in $|P|$ is only constant, and thus we have $|P| = \Oh{n}$.
\end{proof}

\noindent
We now show the main result of this section.

\begin{theorem}
	\label{th:seq-super-fast}
	Given a seed graph $G_0=(V_0, E_0)$ with $n_0 = |V_0|$ nodes, \nlpagen adds $N$ new nodes to $G_0$ in expected time $\Oh{n_0 + N}$.
\end{theorem}

\begin{proof}
	The initialization in line 1 of Algorithm \ref{algo:nlpagen} takes at most time $\Oh{|P_0|} = \Oh{n_0}$ (see Lemma \ref{th:seq-linear-space}).
	The outer loop in lines $2-13$ terminates after $N$ steps.
	The first inner loop in lines $3-6$ terminates once a node is accepted.
	Reusing definitions of Thm.~\ref{th:seq-super-cor}, the expected number of proposals until a node is accepted in step~$i$ is
	\begin{equation}
		\nonumber
		E[T] = \frac{1}{1 - q} = \frac{W'_{i-1}}{W_{i-1}} = \frac{|P_{i-1}| \max_{v \in V_{i-1}} w_{i-1}(v)}{W_{i-1}}.
	\end{equation}

	Recall that a weight $w(v)$ can increase only if node $v$ is accepted, at which point, we add $v$ to $P$ until $w(v) \leq W / n$.
	Then we have $\max_{v \in V_{i-1}} w_{i-1}(v) \leq W_j / n_j$ where $j < i $ is the last step in which the node with the maximum $w(v)$ was accepted, and we obtain
	\begin{equation}
		\nonumber
		E[T] \leq \frac{W_j}{n_j} \frac{|P_{i-1}|}{W_{i-1}} = \Oh{ \frac{W_j}{n_{j}} \frac{n_{i-1}}{W_{i-1}} }
	\end{equation}
	where the last equality follows from the upper bound on $|P|$ given by Lemma \ref{th:seq-linear-space}.
	
	Now if $W_j / n_j \leq W_{i-1} / n_{i-1}$, then  $E[T] = \Oh{1}$ as desired.
	The other case is $W_j / n_j \geq w(v) > W_{i-1} / n_{i-1}$.
	In this case, we can show that $v$ is sampled again adjusting its weight before $E[T]$ grows too large.
	Note that since $W_{i-1} / n_{i-1} < w(v) \leq W_j / n_j$, there has to be some step $j < k \leq i - 1$ with $W_{k} / n_{k} < w(v) < W_{k-1} / n_{k-1}$.
	Since $w(v) = \dv{v}^{\alpha} / c(v) > W_{k} / n_{k}$ implies that $\dv{v}^{\alpha} > W_{k} / n_k$, the probability of sampling $v$ in step $k$ is at least $1 / n_k$.
	In addition, since $W / n$ has to decrease for $E[T]$ to increase, we still have $w(v) > W_{l} / n_{l}$ in some step $l > k$ with $n_{l} = C n_k$ for some $C > 1$.
	This implies that the probability of sampling $v$ in step $l$ is at least $1 / n_{l} = 1 / C n_{k}$.
	We now examine the expected number of times that $v$ is sampled between step $k$ and $l$ and find
	\begin{equation}
		\nonumber
		\frac{1}{n_k} + \frac{1}{n_{k + 1}} + \dots + \frac{1}{C n_{k}} = \mathbf{H}_{C n_{k}} - \mathbf{H}_{n_{k - 1}}
	\end{equation}
	where $\mathbf{H}_i$ denotes the $i$-th harmonic number.
	Using $\mathbf{H}_{C n_{k}} - \mathbf{H}_{n_{k - 1}} = \ln(C) + o(1)$, we find that for $C = e$ node $v$ is expected to be sampled once, and thus in expectation, $E[T]$ cannot grow larger than
	\begin{equation}
		\nonumber
		E[T] = \Oh{\frac{e n_{k}}{n_k}} = \Oh{1}.
	\end{equation}
	
	For the second inner loop in lines $10-12$, it suffices to use the bound on $|P|$ given by Lemma \ref{th:seq-linear-space}.
	As the bound gives $|P_N| = \Oh{n_0 + N}$ after $N$ steps, and each iteration of the loop increases the size of $P$ by $1$, the total time spent in this loop is $\Oh{n_0 + N}$.
\end{proof}

\section{Parallel Algorithm}
\label{sec:par-nonlinear-pa}
In this section, we describe an efficient parallelization of the \nlpagen algorithm given in \cref{subsec:seq-nlpagen}.
The parallel algorithm \pnlpagen builds on the following observation.

\begin{myobs}
	\label{obs:weight-increase-small}
	Let $G_0, \dots, G_N$ be a sequence of graphs generated with polynomial preferential attachment for some $\alpha \in \mathbb{R}_{> 0}$ as per \cref{def:pa}, and let $W_i = \sum_{v \in V_i} \dvi{v}{i}^\alpha$ denote the sum of the node weights after step $i$.
	Then, we have
	\begin{equation}
		\nonumber
		\frac{W_{i+1} - W_{i}}{W_{i+1}} = \begin{cases} \Oh{1 / n_i} &\quad \text{if } \alpha \leq 1. \\ \Oh{1 / \Delta_i} &\quad \text{if } \alpha > 1. \end{cases}
	\end{equation}
\end{myobs}

\noindent
\cref{obs:weight-increase-small} suggests that a sample drawn in step $i+1$ is independent from any changes to the distribution caused by step $i$ with a rather large probability.
Extending this principle to a batch of $\sqrt{n_i}$ or $\sqrt{\Delta_i}$ samples, we see that all samples in the batch are independent from changes to the distribution caused within the same batch with a non-vanishing probability.
Thus, we may draw all samples in a batch independently in parallel until the first dependent sample, which gives an efficient parallelization for $\nproc = \Theta(\sqrt{n})$ processors if $\alpha \leq 1$, and $\nproc = \Theta(\sqrt{\Delta_N})$ if $\alpha > 1$.

\subsection{ParPolyPA}
\label{subsec:par-nlpagen}

\begin{algorithm2e}[t]
	\KwData{Seed graph $G = (V, E)$, with $n_0 = |V|$, $m_0 = |E|$, requested nodes $N$, $\alpha > 0$}
	\ForPar(\hfill \CommentSty{// Init $P$}){$1 \leq p \leq \nproc$}{
		$P_p \gets [v_{\frac{(p - 1) n_0}{\nproc}}, \dots, \underbrace{v_i, \dots, v_i}_{\lceil \dvi{v_i}{0}^\alpha n_0 / W_0 \rceil}, \dots, v_{\frac{p n_0}{\nproc}}]$
	}
	$s \gets n_0$\;
	\While(\hfill \CommentSty{// Start new batch}){$s < N + 1$}{
		$l \gets N + 1$\;
		\ForPar{$1 \leq p \leq \nproc$}{
			$H \gets []$\tcp*{Phase 1}
			\For{$i \in [s + p, s + p + \nproc, \dots, l]$}{
				\textbf{with prob.} $\frac{W_i' - W_{s}}{W_{i}'}$, $l \gets \min \{i, l\}$\;
				Sample $h$ with $P_1(h) = \dvi{h}{s}^\alpha / W_s$\;
				$H \gets H + [(v_{n_0 + i}, h, i)]$\;
			}
			\textbf{barrier}\tcp*{Phase 2}
			\For{$(v, h, i) \in H$ \normalfont{where} $i < l$}{
				$V \gets V \cup \{v\}, E \gets E \cup \{\{v, h\}\}$\;
				$P_p \gets P_p + [v]$\;
				\While{$w(h) > W_s / n_s$}{
					$P_p \gets P_p + [h]$\;
				}
			}
			\textbf{barrier}\tcp*{Phase 3}
			\If{$p$\normalfont{ responsible}}{
				\textbf{with prob.} $\frac{W_{l} - W_s}{W_{l}' - W_s}$, Sample $h$ with $P_2(h) = (\dvi{h}{l}^\alpha - \dvi{h}{s}^\alpha)/ (W_l - W_s)$ \textbf{else} with $P_3(h) = \dvi{h}{l}^\alpha / W_l$\;
				$V \gets V \cup \{v_l\}, E \gets E \cup \{\{v_l, h\}\}$\;
				$P_p \gets P_p + [v_l]$\;
				\While{$w(h) > W_l / n_l$}{
					$P_p \gets P_p + [h]$\;
				}
			}
		}
		$s \gets l$\;
	}
	Return $G =(V, E)$\;
	\caption{\pnlpagen}
	\label{algo:pnlpagen}
\end{algorithm2e}

We now describe the parallel algorithm in detail, see also Algorithm \ref{algo:pnlpagen}.
First, let $\nproc$ denote the number of PUs (processors), and let $1 \leq p \leq \nproc$ denote the $p$-th PU.
Then, PU $p$ adds all new nodes at positions $n_0 + k \nproc + p$ and all new edges at positions $m_0 + k \nproc + p$ where $0 \leq k \leq N / \nproc$\footnote{For simplicity, we assume that $\nproc$ divides $n_0$ and $N$.}.
Before any samples are drawn, PU $p$ receives an $n_0 / \nproc$ share of the nodes in the seed graph, and initializes a proposal list $P_p$ for its nodes as described in \cref{subsec:seq-nlpagen}.

Next, all PUs enter the sampling stage.
Sampling is done in batches, where each batch ends if a dependent sample is found.
We indicate this event by an atomic variable $l$ that is initially set to $l \gets N + 1$, but may be updated with decreasing indices of dependent samples during the batch.
We also let $s$ denote the index of the first sample in each batch.
Each batch consists of three phases.
Note that the phases are synchronized, e.g. no PU may proceed to the next phase until all PUs have finished the current phase.

In the first phase, each PU draws its samples independently and stores them in a list of its hosts for this batch, but does not yet add any nodes or edges to the graph.
Before sample $i$, PU $p$ flips a biased coin that comes up heads with probability $W_{s} / W_{i}'$ where
\begin{equation}\label{eq:W-upper-bound}
	\nonumber
	W_i' = \begin{cases} W_s + 2 (i - s) & \text{if } \alpha \leq 1 \\ W_s + 2 \left({(\Delta_s + i - s)}^{\alpha} - {\Delta_s}^{\alpha}\right) & \text{if } \alpha > 1 \end{cases}
\end{equation}
gives an upper bound on $W_i$.
If the result is heads, then we know that the sample has to come from the old distribution, \ie node $h$ should be sampled with probability $P_1(h) = \dvi{h}{s}^\alpha / W_s$.
To this end, the PU draws two uniform random indices $r \in \{1, \dots, \nproc\}$ and $c \in \{1, \dots, \mathcal{S}\}$ where $\mathcal{S} = \max_p |P_p|$, and requests the $c$-th entry of list $P_r$ of PU $r$.
Note that it is possible for this request to fail if $c > |P_r|$, and in this case, two new indices are sampled.
Once a candidate node $h$ is found, $h$ is accepted or rejected as in the sequential case, and once a sample is accepted, it is added to the list of hosts.
If however, the result is tails, then the sample may have to come from new distribution, so the batch has to end before drawing this sample.
Therefore, the PU atomically checks variable $l$, and if $i < l$, updates the lower bound by setting $l \gets i$.
The PU terminates the first phase if $i \geq l$ or all nodes have been processed; otherwise it continues with its next node.

In the second phase, each PU adds its new nodes and sampled edges to the graph, and updates its list $P$ by adding its nodes, and any increase in the counts of its hosts caused by its edges.

In the third phase, the PU that found the first dependent sample $l$ draws this sample (if any).
Recall that we overestimated the probability that the sample was dependent by using the upper bound $W_l'$.
To correct for this, we first flip a coin that comes up heads with probability $(W_{l} - W_s)  / (W_{l}' - W_s)$.
If the result is heads, the responsible PU draws the sample only from the weight added in the batch, i.e. node $h$ is sampled with probability $P_2(h) = (\dvi{h}{l}^\alpha - \dvi{h}{s}^\alpha) / (W_l - W_s)$.
To this end, it selects a candidate node from one of the lists $P$, but only considers positions added during the batch.
Consequently, a candidate is accepted with probability proportional to the increase in its weight divided by the increase in its count.
Otherwise, if the result is tails, the sample is drawn from the new distribution with probability $P_3(h) = \dvi{h}{l}^\alpha / W_l$.
Once the PU sampled a host, it adds the node and edge to the graph and updates its list $P$.
Then, all PUs enter the next batch, or exit, if all $N$ samples have been drawn.

We now prove the correctness of the necessary modifications to obtain \pnlpagen from \nlpagen.

\begin{theorem}
	\label{th:par-cor}
	\pnlpagen samples host~$h \in V$ with probability $\dv{h}^\alpha / W$.
\end{theorem}

\begin{proof}
	We distinguish three cases to sample node $h$.
	Either (1) $h$ is sampled as independent sample with probability $P_1(h) = \dvi{h}{s}^\alpha / W_s$, or (2) as dependent sample with probability $P_2(h) = (\dvi{h}{l}^\alpha - \dvi{h}{s}^\alpha) / (W_l - W_s)$, or (3) as sample from the new distribution with probability $P_3(h) = \dvi{h}{l}^\alpha / W_l$.
	Thus, the overall probability of sampling $h$ is given by
	\begin{align}
		\nonumber
		P(h) &= \frac{W_s}{W_l'} P_1(h) + \frac{W_l - W_s}{W_l'} P_2(h) + \frac{W_l' - W_l}{W_l'} P_3(h)
		\\
		\nonumber
		&= \frac{\dvi{h}{s}^\alpha}{W_l'} + \frac{\dvi{h}{l}^\alpha - \dvi{h}{s}^\alpha}{W_l'} + \frac{W_l' - W_l}{W_l'} \frac{\dvi{h}{l}^\alpha}{W_l}
		\\
		\nonumber
		&= \frac{\dvi{h}{l}^\alpha}{W_l'} + \frac{W_l' - W_l}{W_l'} \frac{\dvi{h}{l}^\alpha}{W_l}
		\\
		\nonumber
		&= \frac{\dvi{h}{l}^\alpha}{W_l}.
	\end{align}
	It only remains to show that $W_l'$ is an upper bound on $W_l$.
	We first consider the case where $\alpha \leq 1$.
	Observe that in each step, $W_s$ increases by $1$ for the new node and by $(\dv{h} + 1)^\alpha - \dv{h}^\alpha \leq 1$ for the host $h$.
	Thus, the overall increase after $l - s$ steps is at most $2 (l - s)$, and $W_l \leq W_s + 2 (l - s) = W_l'$ as claimed.
	In the other case, where $\alpha > 1$, $W_s$ similarly increases by $1$ for the new node and by $(\dv{h} + 1)^\alpha - \dv{h}^\alpha$ for the host $h$.
	In addition, it is easy to verify that $(\dv{h} + 1)^\alpha - \dv{h}^\alpha$ is maximized if $\dv{h} = \Delta$.
	Thus, we have
	\begin{align}
		\nonumber
		W_l &\leq W_s + l - s + (\Delta_s + l - s)^\alpha - \Delta_s^\alpha 
		\\
		\nonumber
		&\leq W_s + 2 ((\Delta_s + l - s)^\alpha - \Delta_s^\alpha) 
		\\
		\nonumber
		&= W_l'
	\end{align}
	as claimed.
\end{proof}

The following theorem shows that \pnlpagen yields a near-optimal speed-up if $\alpha \leq 1$ or $\alpha > 1$ and $N$ is large.

\begin{theorem}
\label{th:par-fast}
	Given a seed graph $G_0=(V_0, E_0)$ with $n_0 = |V_0|$, \pnlpagen adds $N$ new nodes to $G_0$ in expected time $\Oh{(\sqrt{N} + N / \nproc) \log \nproc}$ if $\alpha \leq 1$ or $\alpha > 1$ and $N = \omega(n_0)$.
\end{theorem}

\begin{proof}
	Computing $W_0$ and initializing the lists $P_1, \dots, P_\nproc$ in parallel takes time $\Oh{n_0 / \nproc \log(\nproc)}$.

	In the following, we show that the remainder of \pnlpagen can be implemented to process a batch of length~$R$ in time $\Oh{(1 + R/\nproc) \log \nproc}$.
	By \cref{obs:weight-increase-small} and $\Delta_N = \Theta(n)$ if $N = \omega(n_0)$, we expect $E[R] = \Theta(\sqrt{n_i})$ resulting in $\Oh{\sqrt{N}}$ batches which, in combination, establishes the claim.
	
	In the first phase, all samples are drawn independently and affect only each PU's local state.
	The only concurrent writes happen to the global variable $l$ for which we use \emph{minimum} as conflict resolution.
	Thus, the first phase can be executed in time $\Oh{(1 + R/\nproc) \log \nproc}$.
	
	The runtime of the second phase is dominated by the concurrent writes to the shared data structuring maintaining the counts $c(\cdot)$ of the proposal list.
	Parallel writes to the same counter are summed up.
	There at most $\Theta(R)$ such updates, accounting for $\Oh{(1 + R/\nproc) \log \nproc}$ time per batch.
	
	Finally, in the third phase, only one PU draws one sample which takes constant time.
\end{proof}

\section{I/O-Efficient Algorithm}
\label{sec:external-generalized-pa}
\newcommand{\HostReq}[3]{\textsf{HostReq}$\langle #1, #2, #3 \rangle$}
\newcommand{\ExMsg}[3]{\textsf{ExMsg}$\langle #1, #2, #3 \rangle$}
\newcommand{\PQU}{\ensuremath{\textsf{PQ}_U}\xspace}
\newcommand{\PQM}{\ensuremath{\textsf{PQ}_M}\xspace}

In this section, we extend the I/O-efficient algorithm of \cite{DBLP:conf/alenex/MeyerP16} to the general case.
The algorithm transfers the main idea of \cite{batagelj2005efficient} to the external memory setting.
Rather than reading from random positions of the edge list, it emulates the same process by precomputing all necessary read operations and sorts them by the memory address they are read from.
As the algorithm produces the edge list monotonously moving from beginning to end, it scans through the sorted read requests and forwards still cached values to the corresponding target positions using an I/O-efficient priority-queue.

In order to extend the algorithm to the general case, we split the sampling of hosts into a two-step process, see also Algorithm \ref{algo:extgeneralizedpa}.
First, for each new node $v_{n_0 + i}$ we only sample the degrees of the $\ell$ different hosts and collectively save this information for the second phase, \eg~that node $v_{n_0 + i}$ requested a host with degree $d$.
By postponing the actual sampling of the hosts, we can bulk all nodes with the same degree and therefore distribute them I/O-efficiently to the incoming nodes.

\begin{algorithm2e}[t]
	\KwData{Seed graph $G = (V, E)$, with $n_0 = |V|$, $m_0 = |E|$, requested nodes $N$}
	\ForEach(\hfill \CommentSty{// Init}){$v \in V$}{
		\PQM.push(\ExMsg{\deg_{0}(v)}{0}{v}) \\
		$c_{\deg_{0}(v)} \leftarrow c_{\deg_{0}(v)} + 1$
	}
	\ForEach(\hfill \CommentSty{// Phase 1}){$i \in [1, \ldots, N]$}{
		\PQM.push(\ExMsg{\ell}{\ell(i+1)}{v_{n_0 + i}}) \\
		\ForEach{$j \in [1, \ldots, \ell]$}{
			\Repeat{accepted}{
				Sample degree $d := \deg(h)$ of host $h$ proportionally to $c_df(d)$ \\
				Accept with prob.~$(c_{d} {-} s_{d}) / c_{d}$
			}
			$s_{\deg(h)} \leftarrow s_{\deg(h)} + 1$ \\
			Store \HostReq{v_{n_0+i}}{j}{\deg(h)}
		}
		Update counters $c_{d}$ with changes $s_{d}$
	}
	Sort requests ascendingly by (degree, node)\\
	\For(\hfill \CommentSty{// Phase 2}){\normalfont increasing degree $d$}{
		$R_{\min} \leftarrow 0, R_{\max} \leftarrow 1$ \\
		\ForEach{\normalfont \HostReq{v_{n_0 + i}}{j}{d_j} with $d_j{=}d$}{
			$t \leftarrow \ell i + j$ \\
			\While{\normalfont \ExMsg{d'}{t'}{v} $\leftarrow \PQM$.top() where $d' = d$ and $t' < t$}{
				$R_v \leftarrow$ uniformly from $[R_{\min}, R_{\max}]$ \\
				\PQU.push($R_v, v$) \\
				\PQM.pop()
			}
			$(R_u, u) \leftarrow \PQU$.pop()\\
			$R_{\min} \leftarrow R_u$\\
			$E \leftarrow E \cup \{ \{v_{n_0 + i}, u \} \}$
		}
		Empty \PQU
	}
	\caption{\textbf{\emgenpa}}
	\label{algo:extgeneralizedpa}
\end{algorithm2e}

As the first phase only samples node degrees, it suffices to group nodes with the same degree $d$ and represent each group by their counter $c_d$ and weight $c_d \cdot f(d)$.
Therefore, initially each seed node $v \in V_0$ contributes weight $f(\deg_0(v))$ to the overall weight of its corresponding group.
A host request for degree $d$ is then sampled proportionally to $c_d \cdot f(d)$.
To reflect the actual generation process, we remember the number of sampled hosts $s_{d}$ from degree $d$ arising from a new node $v_{n_0+i}$.
This is necessary, as any subsequent host request needs to seek a different node to faithfully represent the model.
By adding rejection sampling we correct the probability distribution a posteriori, \ie~if degree $d$ is sampled we finally accept with probability $(c_d - s_d)/c_d$ and restart otherwise.
After generating all $\ell$ host requests we update the corresponding counters $(c_d, c_{d+1})$ to $(c_d - s_d, c_{d+1} + s_d)$ for at most $\ell$ degrees.
While sampling these requests we generate tuples \HostReq{v_{n_0 + i}}{j}{d_j} for $j \in \{1, \ldots, \ell\}$ representing that node $v_{n_0 + i}$ requested as $j$-th host a node with degree $d_j$.

After collecting all host requests, we sort all requests by host degree first and new node second.
Subsequently in the second phase, we fulfill each request for a host of degree $d$ by uniformly sampling from the set of existing nodes with degree $d$ at that time using two I/O-efficient priority-queues \PQU and \PQM employing standard external memory techniques \cite{DBLPedit:conf/dagstuhl/MaheshwariZ02}.
While \PQU is simply used as a means to retrieve uniform samples of its currently held messages, \PQM is used to gather all matching nodes.
To initialize, given the seed graph $G_0$ we insert for each node $v \in V_0$ a message \ExMsg{\deg_0(v)}{0}{v} into \PQM reflecting the information that at time $t = 0$ node $v$ has degree $\deg_0(v)$ in $G_0$.
Similarly, we insert messages \ExMsg{\ell}{\ell (i + 1)}{v_{n_0 + i}} into \PQM for all $i \in \{1, \ldots, N\}$, hinting that at time $\ell (i + 1)$, after its addition, node $v_{n_0 + i}$ indeed has $\ell$ neighbors.
We then process requests ascendingly by degree.

When processing host request \HostReq{v_{n_0 + i}}{j}{d_j} we compute $t = \ell i + j$ and push all vertices of messages from \PQM with degree $d_j$ and time $t' < t$ into \PQU.
More concretely, when processing all requests with degree $d$ it is necessary to keep two values $R_{\min}$ and $R_{\max}$ and insert nodes into \PQU with a weight drawn uniformly at random from $[R_{\min}, R_{\max}]$.
After all suitable nodes have been inserted into \PQU, we pop the node with smallest weight $R$ and connect $v_{n_0 + i}$ to it\footnote{Due to interchangeability, all inserted nodes have equal probability to have minimum weight and are thus sampled uniformly at random.}.
Now that all remaining nodes in \PQU have a weight of at least $R$, we update $R_{\min} \leftarrow R$ to preserve uniformity for any following node.
Essentially, for a request targeted at time $t$ and degree $d$ we provide all nodes that exist as nodes with degree $d$ up to time $t$ and uniformly sample from them.
After fulfilling the request with node $u$, we forward $u$ as a potential partner to requests for hosts with degree $d + 1$ by adding a message \ExMsg{d + 1}{\ell(i + 1)}{u} into \PQM.
All remaining unmatched nodes in \PQU will stay unmatched, hence \PQU is simply emptied and the algorithm proceeds to requests for the next larger degree.

By lazily resolving the actual hosts in the second step and only sampling the degrees in the first, the algorithm only needs to keep the set of currently existing degree groups in internal memory, \ie~the respective degree and its multiplicity in the current graph.
Note that a graph represented by its set of unique degrees $D$ has $\Omega(|D|^2)$ edges which even under pessimistic assumptions amounts to an output graph with more than 1 PB of size \cite{DBLPedit:journals/jea/HamannMPTW18}.

\begin{lemma}
	Processing the host requests sorted ascendingly by degree first and new node second correctly retains the output distribution.
\end{lemma}
\begin{proof}
	Let $M(d, t)$ be the set of nodes with degree $d$ after time $t$ where $M(d, 0)$ is given by the nodes of the seed graph $G_0$ with degree $d$.
	Processing a host request \HostReq{v_{n_0 + i}}{j}{d} of node $v$ for degree $d$ at time $t = \ell i + j$ uniformly samples a node $u$ of $M(d, t - 1)$ and forwards it to the next degree group reflecting the equalities $M(d, t) = M(d, t - 1) \setminus \{ u \}$ and $M(d + 1, t) = M(d + 1, t - 1) \cup \{ u \}$. 
	
	Therefore, when considering two requests \HostReq{v_{n_0 + a}}{j}{d} and \HostReq{v_{n_0 + b}}{j'}{d'} where the first is generated before the second, it is clear that the latter can only depend on the former when $d \le d'$, implying a DAG of dependencies of the degree requests.
	In particular, it is only necessary to consider direct dependencies, see \cref{fig:degree-requests-dependencies} for an example.
	
	Processing the requests is correct as long as it conforms to the DAG of dependencies, \ie~is equivalent to any topological ordering of the requests.
	Naturally, processing requests ascendingly by time, is therefore correct.
	However, the order given by ascendingly sorting by degree first and time second also corresponds to a topological ordering proving the claim.
	
	In \cref{fig:degree-requests-dependencies} both approaches are easily visualized.
	Processing by time corresponds to fulfilling the requests from left to right while processing by degree first and time second corresponds to fulfilling the requests from top to bottom, where early requests are prioritized if the degree is matching.
\end{proof}

\begin{theorem}
	Given a seed graph $G_0 = (V_0, E_0)$ with $n_0 = |V_0|$, \emgenpa adds $m$ new edges to $G_0$ using $\Oh{\sort{n_0 + m}}$ I/Os if the maximum number of unique degrees fits into internal memory.
\end{theorem}
\begin{proof}
	The first phase produces $\Theta(m)$ many host requests incurring $\Theta(\scan{m})$ I/Os where the sampling can be done efficiently using a dynamic decision tree.
	Sorting the requests then takes $\Oh{\sort{m}}$ I/Os.
	
	In the second phase, each request is read in ascending fashion requiring a single scan, incurring $\Theta(\scan{m})$ I/Os.
	Fulfilling the host requests is done iteratively for increasing target degree.
	If a node is matched with a request for degree $d$, it is forwarded in time to requests for degree $d+1$ or cleared otherwise.
	Thus, each node $v$ in the output graph is inserted at least once and reinserted at most $\Oh{\deg_N(v) - \deg_0(v)}$ times into the priority-queues.
	Hence, in total $\Oh{n_0 + m}$ messages are produced incurring $\Oh{\sort{n_0 + m}}$ I/Os when implementing the priority-queues using Buffer Trees as the underlying data structure~\cite{DBLP:journals/algorithmica/Arge03}.
\end{proof}

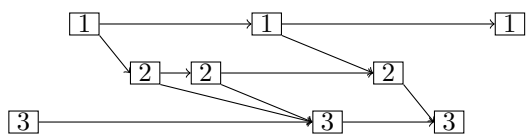
\begin{figure}[t]
	\centering
	\begin{tikzpicture}[
	block/.style = {
		draw, 
		fill=white, 
		rectangle, 
		minimum width=width("12") + 1pt,
		inner sep=1pt
	},
	bbedge/.style = {
		->,
		draw
	}]
	\def\x{.8};
	\def\y{.65};
	\node[block] (t0) at (0*\x, 0*\y) {3};
	\node[block] (t8) at (8*\x, 2*\y) {1};
	
	\path [use as bounding box] (t0.south west) -| (t8.north east) -| (t0.south west);
	
	\node[block] (t1) at (1*\x, 2*\y) {1};
	\node[block] (t2) at (2*\x, 1*\y) {2};
	\node[block] (t3) at (3*\x, 1*\y) {2};
	\node[block] (t4) at (4*\x, 2*\y) {1};
	\node[block] (t5) at (5*\x, 0*\y) {3};
	\node[block] (t6) at (6*\x, 1*\y) {2};
	\node[block] (t7) at (7*\x, 0*\y) {3};
		
	\path[bbedge] (t0.east) -- (t5.west);
	
	\path[bbedge] (t1.south east) -- (t2.west);
	\path[bbedge] (t1.east) -- (t4.west);
	
	\path[bbedge] (t2.east) -- (t3.west);
	\path[bbedge] (t2.south east) -- (t5.west);
	
	\path[bbedge] (t3.south east) -- (t5.west);
	\path[bbedge] (t3.east) -- (t6.west);
	
	\path[bbedge] (t4.east) -- (t8.west);
	\path[bbedge] (t4.south east) -- (t6.west);
	
	\path[bbedge] (t5.east) -- (t7.west);
	
	\path[bbedge] (t6.south east) -- (t7.west);
\end{tikzpicture}
	\caption{
		Dependency graph of direct dependencies where the degrees $(3, 1, 2, 2, 1, 3, 2, 3, 1)$ are requested in order from left to right.
	}
	\label{fig:degree-requests-dependencies}
\end{figure}

\section{Implementations}
\label{sec:implementations}
\subsection{Implementation for Internal Memory}
\label{subsec:impl-internal-memory}
We implement \nlpagen, \pnlpagen, and \dyngen (see below) in the programming language  \textsc{Rust}\footnote{\url{https://rust-lang.org}; performance roughly on par with C.} and almost exclusively use the \textsc{safe} language subset and avoid hardware-specific features.

For comparison with the state of the art, we consider the generator \dyngen which relies on the first algorithm proposed in \cite{DBLP:journals/mst/MatiasVN03}.
This data structure samples weighted items from a universe of size $N$ in expected time $\Oh{\log^* N}$ and supports updates in time $\Oh{2^{\log^* N}}$.
While the authors propose further asymptotical improvements, preliminary experiments suggest that these translate into slower implementations.
We consider the final implementation well tuned and added a few asymptotically sub-optimal changes (\eg exploiting the finite precision of float point numbers and removing the hash maps originally used) that improve the practical performance significantly.
The code is designed as a standalone \texttt{crate}\footnote{Roughly speaking, the Rust equivalent of a software library.} and will be made independently available to the Rust-ecosystem for general sampling problems.

All implementations share a code base for common tasks.
Non-integer computations are based on double-precision floating-point arithmetic.
Repeated expensive operations (\eg evaluating $d^\alpha$ for small degrees $d$) are memoized.
The sampling of $\ell$ \emph{different} hosts per new node is supported by rejecting repeated hosts.
Additionally, we investigate \dyngenrem, which temporarily removes hosts from the data structure. 

Our \pnlpagen implementation uses parallel threads operating on a shared memory.
Synchronization is implemented via hurdles barriers\footnote{\url{https://github.com/jonhoo/hurdles}}.
All concurrently updated values are accessed either via fetch-and-add or compare-exchange primitives using the acquire-release semantics~\cite{c++2011iso}.
In contrast to the description in Algorithm~\ref{algo:pnlpagen}, the code uses a single proposal list which is implemented as a contiguous vector.
To avoid overheads and false-sharing, each threads reserves small blocks to write.
It is also straight-forward to merge the third phase with the first phase of the next batch.
This allows us to reduce the number of barriers required to two.

Observe that \nlpagen and \pnlpagen use contiguous node indices and that each connected node has at least one entry in the proposal list.
This allows us to store the first entry of each node only implicitly, at least halving the memory size and number of access to the proposal list~$P$ (see \cref{fig:proposal-list-size}).

\subsection{Implementation for External Memory}
We implement \emgenpa in \textsc{C++} using the STXXL\footnote{We use a fork of STXXL that has been developed ahead of master \url{https://github.com/bingmann/stxxl}.} library~\cite{DBLP:journals/spe/DementievKS08} which offers tuned external memory versions of fundamental operations like scanning and sorting.
It additionally provides many implementations of different external memory data structures.

Since the degrees change incrementally where all incoming nodes have initial degree $\ell$, we use a hybrid decision tree for the first phase.
More concretely, we manage the smallest degrees in the range of $[1, P(\sqrt{n}) ]$ statically and any larger degree dynamically where $P(x)$ is the smallest power of two greater or equal to $x$.
Even for $n_0 + n = 2^{40}$ this amounts to less than 50\,MB of memory for the statically managed degrees.

For the second phase we use the external memory priority-queues provided by STXXL based on~\cite{DBLP:journals/jea/Sanders00}.

\section{Experiments}
\label{sec:experiments}
In this section, we study the previously discussed algorithms empirically.
All generators produce simple graphs according to the preferential attachment model in \cref{def:pa}.
To focus on this process, we compute the results but do not write out the graph to memory.

\subsection{Internal Memory Algorithms}

\begin{figure}
	\centering
	\includegraphics[width=0.5\linewidth]{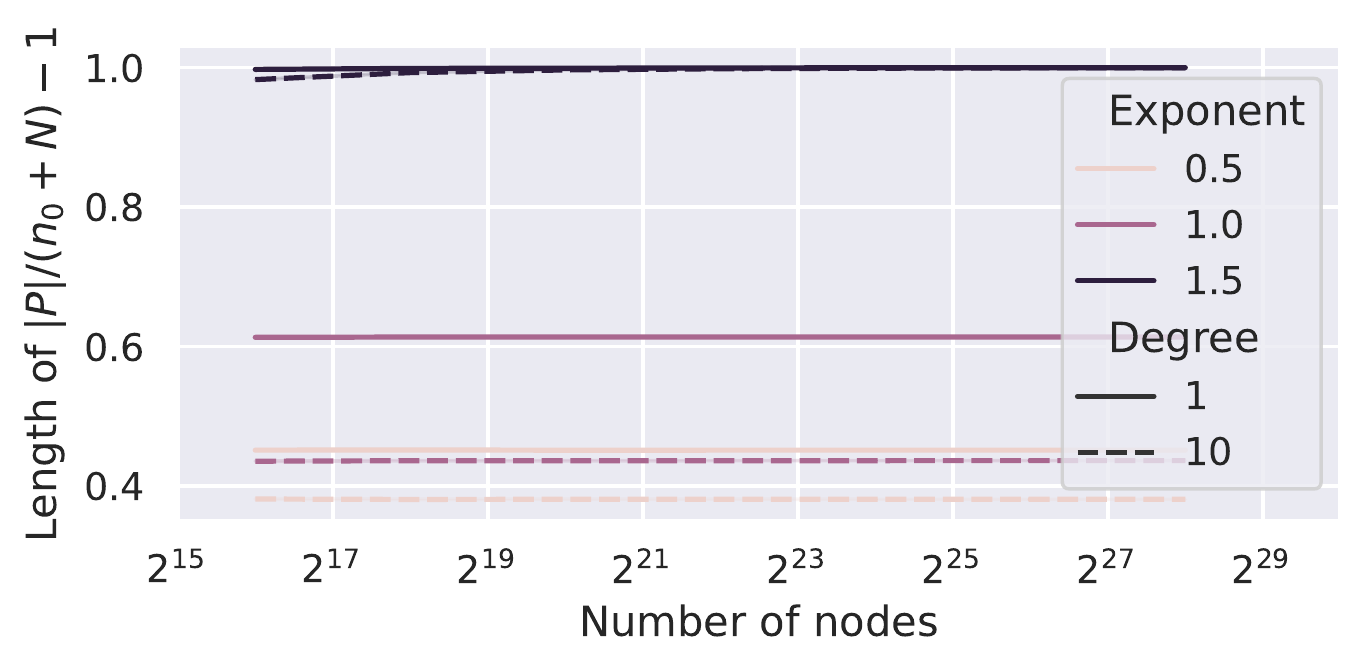}
	\caption{
		Relative length $|P| / (n_0 + N)$ of the proposal list without the $N$ implicitly stored positions.
	}
	\label{fig:proposal-list-size}
	\vspace{-1em}
\end{figure}

\begin{figure}
	\centering
	\includegraphics[height=6cm]{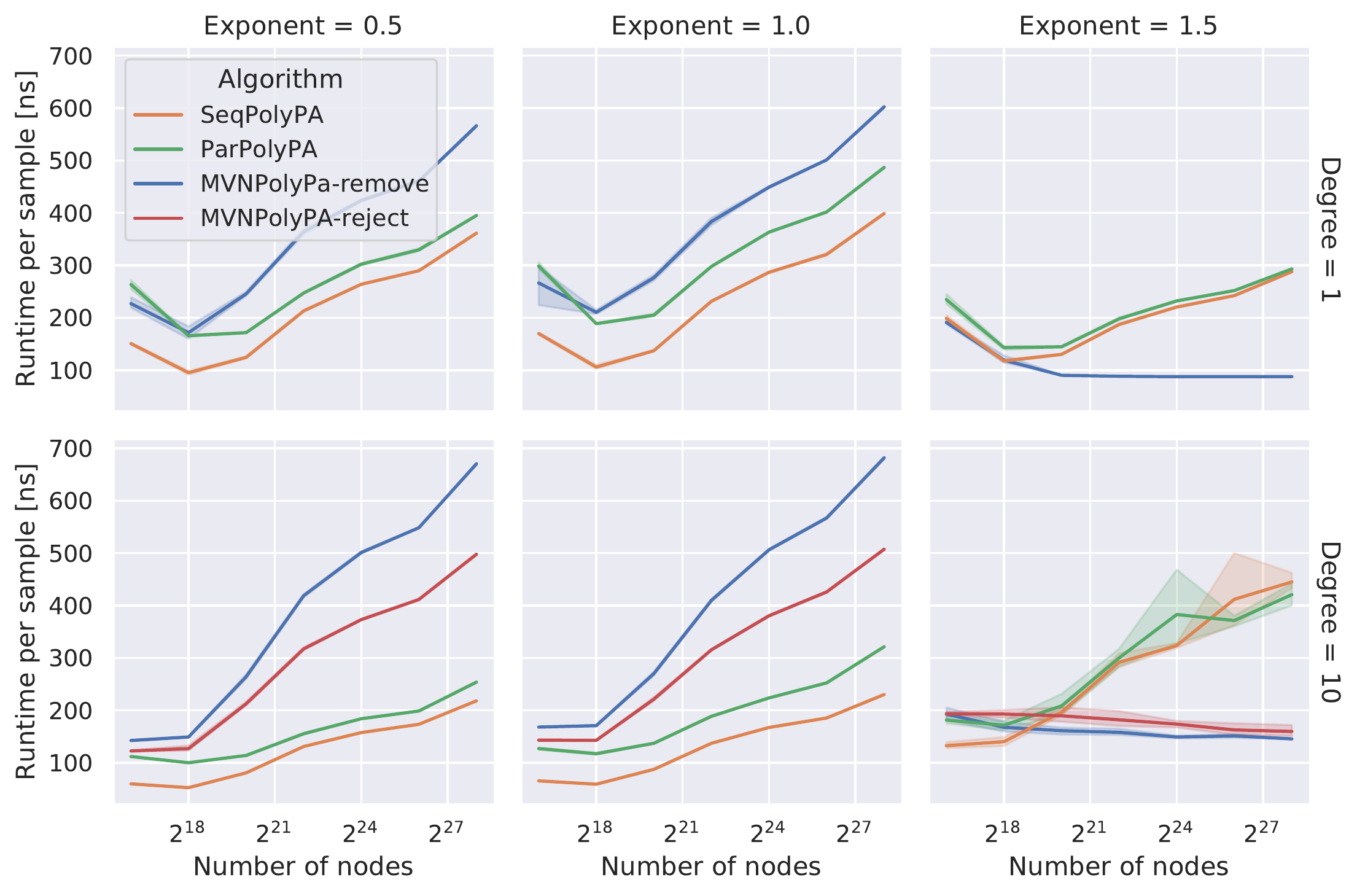}
	\caption{
		Runtime per sample $t_{PA} / (N \ell)$ for different algorithms as function of $N$, $\ell$, and $\alpha$.
	}
	\label{fig:benchmark-seq-scaling}
	\vspace{-1em}
\end{figure}

All internal memory experiments use the implementations described in \cref{subsec:impl-internal-memory} and are build with 
\texttt{rustc 1.66.0-nightly (f83e0266c 2022-10-03)}\footnote{
	See \texttt{Cargo.lock} for the exact versions of all dependencies.
} on an Ubuntu 20.04 machine with an AMD EPYC 7702P processor with 64 cores and 512\,GB of RAM.
To focus our measurements on the sampling phase, we use small $1$-regular seed graphs with $n_0 = 10\ell$ which have a negligible influence on the runtime and the structure of the resulting graph.
We report the wall-time of the preferential-attachment process $t_{PA}$ excluding initial setup costs (\eg seed graph, initial allocation of buffers, et cetera);
this leads to negligible biases between algorithms.

\subsubsection{Sequential performance}
In \cref{subsec:seq-ars}, we bound the expected size of the proposal list to be linear in the graph size.
This analysis is consistent with \cref{fig:proposal-list-size} (Appendix), which reports the proposal size divided by $N$.
We observe no dependency in $N$ over several orders of magnitude and find the proportionality factor to be upper bounded from above by 1 (recall that we store the first entry of each node only implicitly).

\Cref{fig:benchmark-seq-scaling} summaries the scaling behavior of \nlpagen, \pnlpagen, and \dyngen in $N$.
It reports the average time $t_{PA} / (N \ell)$ to obtain a single host for various combinations of $\ell$ and $\alpha$.
Despite near-constant asymptotic predictions, all implementations show a consistent deterioration of performance for larger instances.
We attribute this to unstructured accesses to a growing memory area causing measurable increases in cache misses and back-end stalls.
The additional steep rises in some of \dyngen's plots are due to deeper recursions in the sampling data structure.

While our proposal list-based algorithms are fastest for $\alpha \le 1$, \dyngen performs well for sequential super-linear preferential attachment.
This is due to the expected formation of $\Theta(\ell)$ high-degree nodes in this regime.
These nodes have similar weights and are grouped together by \dyngen which leads to high locality and very few cache misses during sampling.

Observe that it is trivial to hard-code this partition into \nlpagen and \pnlpagen to achieve similar performance.
We opted against it in favor of cleaner measurements that describe the actual performance of the proposal list.
These results are especially representative if the seed graph has a significant contribution to the resulting graph, \ie if $n_0 / N$ is non-vanishing.

For $\ell = 10$, \nlpagen and \pnlpagen need to reject hosts that have been sampled multiple times.
This incurs a small slow-down of less than $1.5\times$.
In this context, \dyngenrem exploits its fully dynamic sampling data structure to temporarily remove hosts (\ie it explicitly samples without replacement).
This is slightly beneficial in the super-linear regime with high locality, while rejection sampling is faster otherwise.

\subsubsection{Parallel performance}
\begin{figure}
	\centering
	\includegraphics[height=5cm]{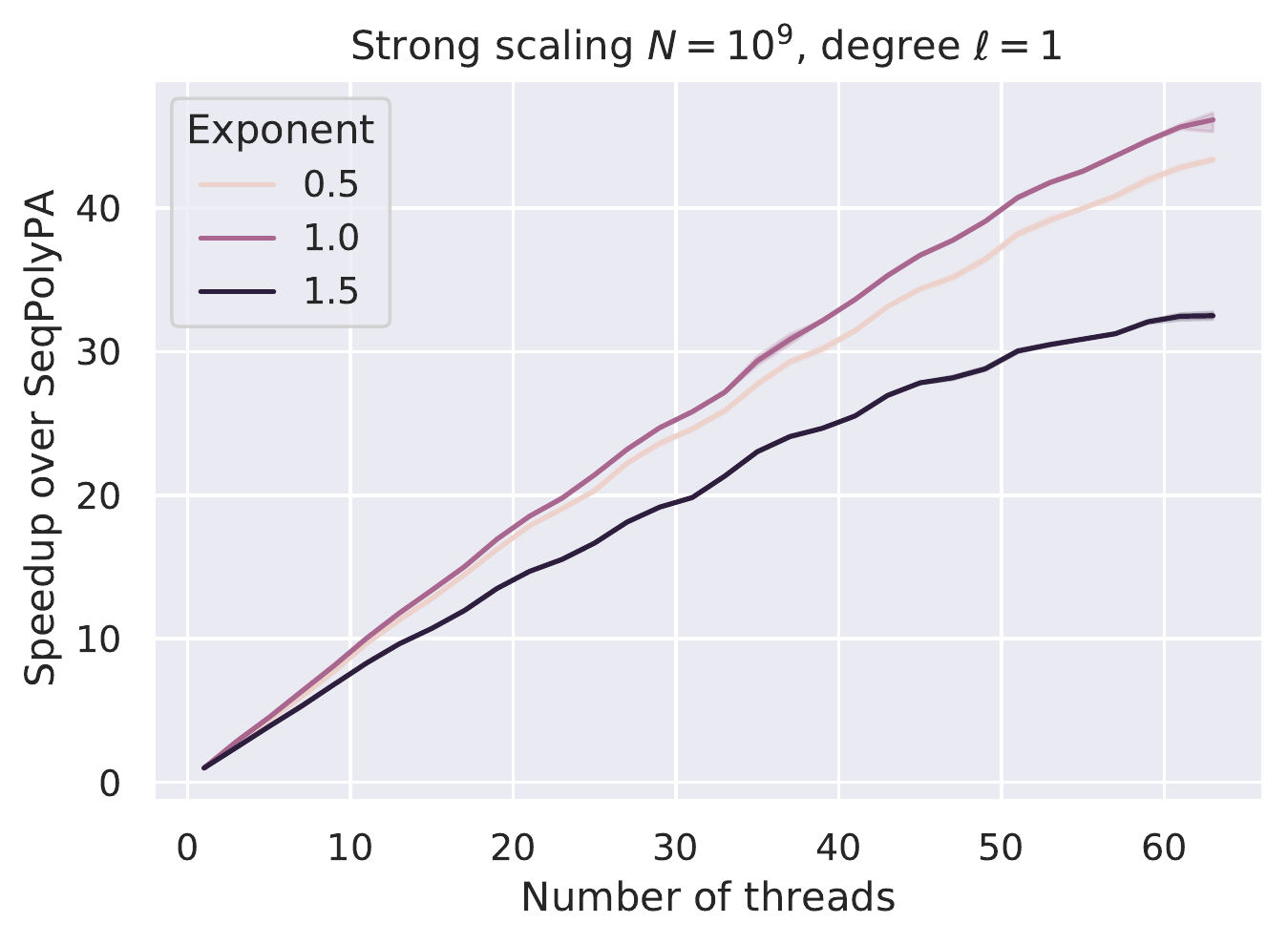}
	\caption{
		Strong scaling of \pnlpagen as speed-up over \nlpagen for $N = 10^9$, $\ell = 1$, and $0.5 \le \alpha \le 1.5$.
	}
	\label{fig:benchmark-par-scaling}
	\vspace{-1em}
\end{figure}

In \cref{obs:weight-increase-small}, we motivate our parallelization by establishing that the number of batches (corresponding to the number of explicit synchronization points) scales as a square root of the graph size / maximum degree.
This is supported by \cref{fig:number-of-batches}, which almost perfectly matches this prediction.

\Cref{fig:benchmark-par-scaling} reports the results for \pnlpagen with $N = 10^9$ as the speed-up over the sequential implementation \nlpagen.
For $\alpha \le 1$, we observe a near linear scaling with a speed-up of up-to $46$ on $63$~threads, dropping to $32$ for $\alpha = 1.5$.

This is despite a comparable number of batches (see \cref{fig:number-of-batches}).
Instead we ---slightly counter-intuitively--- attribute the effect to the increased locality of the super-linear regime, as concurrent updates on high degree nodes lead to more frequent cache-invalidations. 
This can be mitigated using default techniques including randomized update sequences or hierarchical updates.
Similarly to \cite[Sec. 5]{DBLP:conf/esa/BerenbrinkHK0PT20}, one can also merge a small number batches into an epoch, and only update shared data at the epoch's end, leading to a trade-off between increased local overheads and reduced shared updates.

\begin{figure}
	\centering
	\includegraphics[height=3.8cm]{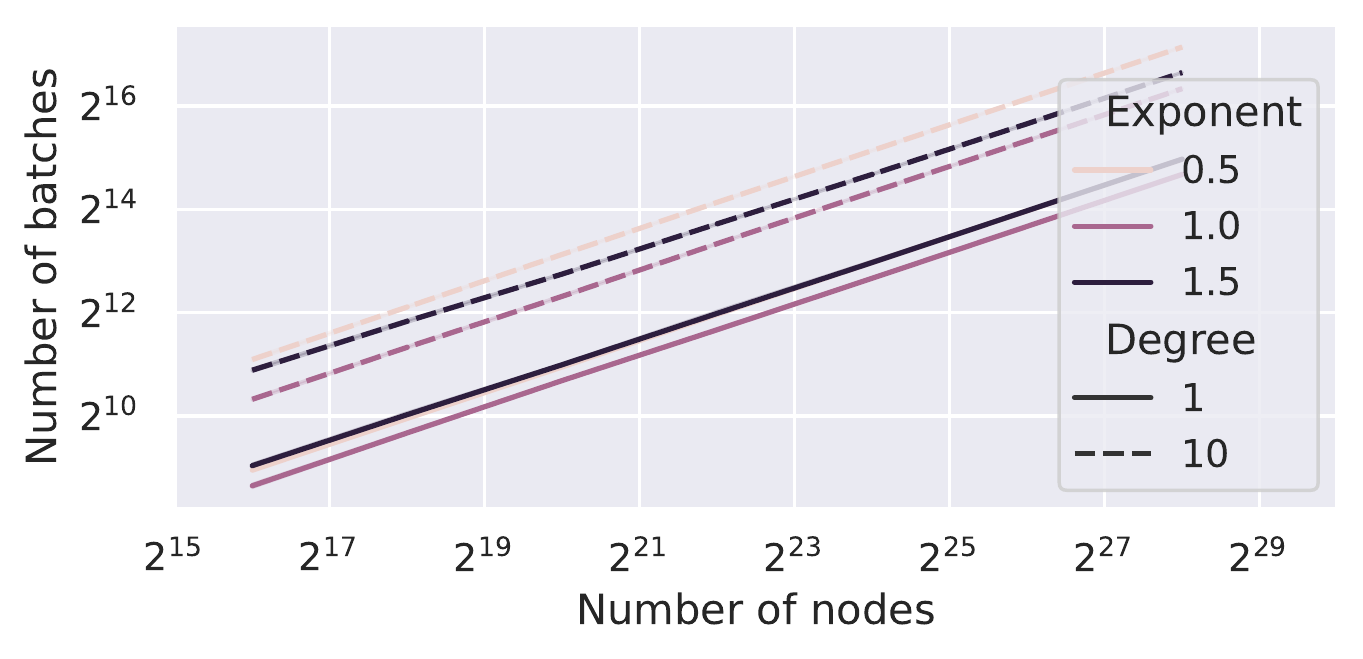}
	\caption{
		The number of batches executed by \pnlpagen for various parameters.
		Observe the slope of $1/2$ almost matching the prediction of \cref{obs:weight-increase-small}.
	}
	\label{fig:number-of-batches}
	\vspace{-1em}
\end{figure}

\subsection{External Memory Algorithms}
In order to assess the computational overhead given by the two-phase sampling, we compare the state-of-the-art sequential external memory algorithm TFP-BA~\cite{DBLP:conf/alenex/MeyerP16} for the linear case to our implementation where we simply set $f(d) = d$.
As an additional reference we consider a fast sequential internal memory implementation of the algorithm of Brandes and Batagelj~\cite{batagelj2005efficient} provided by NetworKit~\cite{DBLP:journals/netsci/StaudtSM16} which we refer to by NK-BA.
Analogously to the experiments presented in~\cite{DBLP:conf/alenex/MeyerP16}, we use a small ring graph with $n_0 = 2 \ell$ nodes as the seed graph for both algorithms and compare their running times for an increasing number of incoming nodes $N$.
The benchmarks are built with GNU g++-9.4 and executed on a machine equipped with an AMD EPYC 7302P processor and 64\,GB RAM running Ubuntu 20.04 using four 500\,GB solid-state disks. 

As illustrated in \cref{fig:extmem-comparison}, the performance of \emgenpa is less than a factor of $2.32$ slower than TFP-BA in the linear case and seems to be independent of $M$.
Furthermore, this discrepancy becomes smaller when including writing the result to disk.
Naturally, for graphs that fit into internal memory NK-BA is the fastest algorithm but becomes infeasible as soon as the edge list exceeds the available internal memory.

To study the influence of $f$ on the running time of \emgenpa, we additionally consider two polynomial models with exponent $\alpha \in \{0.5, 1.5\}$.
In \cref{fig:extmem-models} we see that \emgenpa performs similarly well for all three models.
However, it is noticeable that \emgenpa performs better for models with a larger exponent $\alpha$ due to the increasingly more skewed degree distributions.

\begin{figure}[t]
	\centering
	\includegraphics[height=3.8cm]{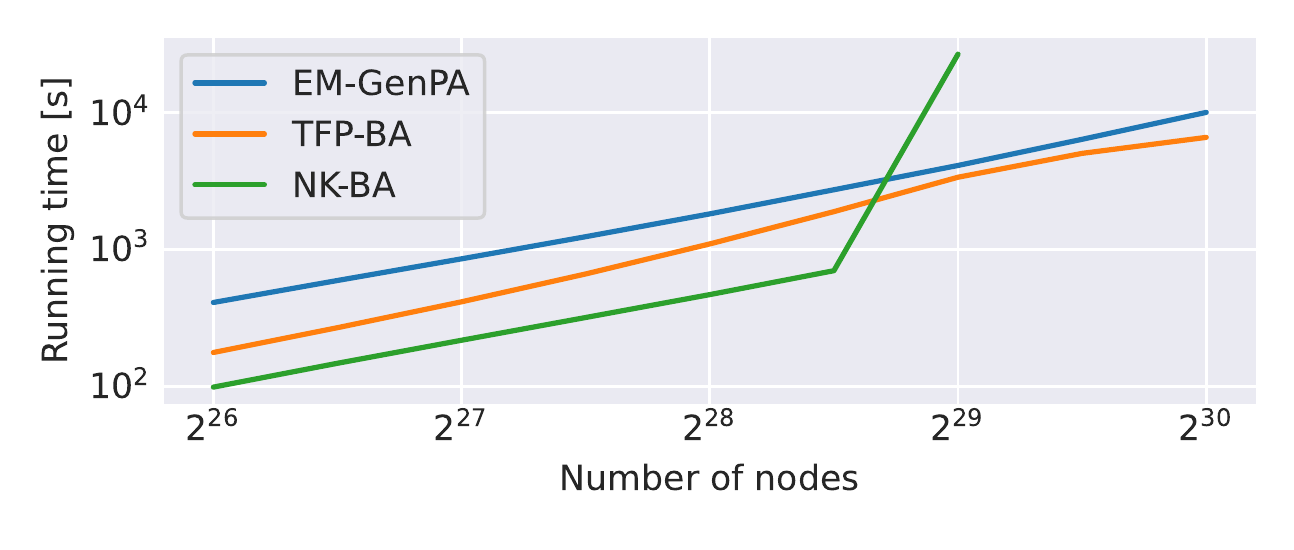}
	\caption{
		Running times of TFP-BA, NK-BA and \emgenpa for $\ell = 10$ and increasing $N$.
	}
	\label{fig:extmem-comparison}
\end{figure}

\begin{figure}[t]
	\centering
	\includegraphics[height=3.8cm]{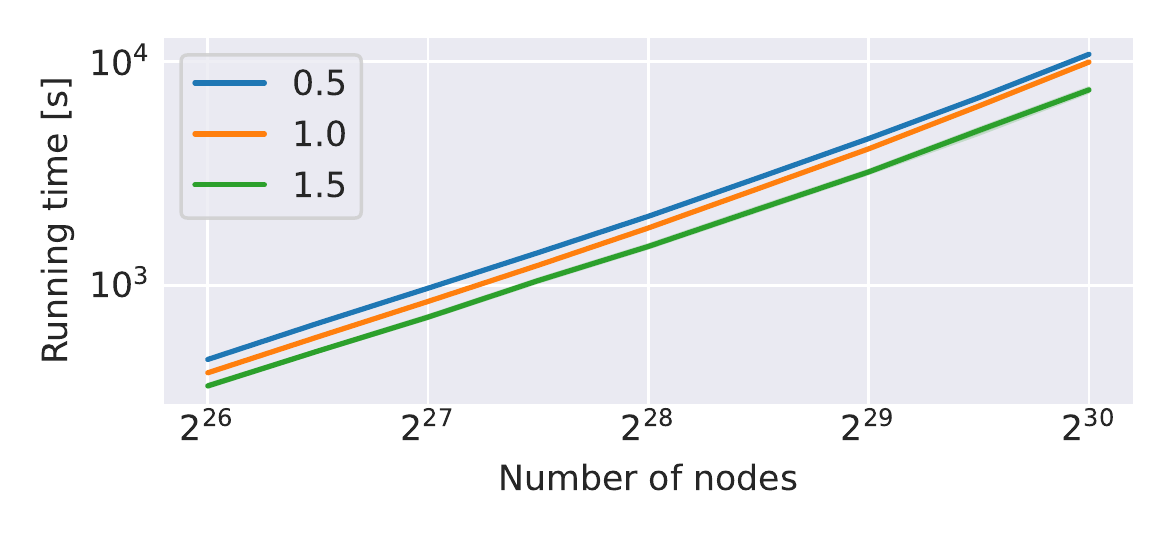}
	\caption{
		Running times of \emgenpa for $\ell = 10, \alpha \in \{0.5, 1.0, 1.5\}$ and increasing $N$.
	}
	\label{fig:extmem-models}
\end{figure}

\section{Conclusions}
\label{sec:conclusions}
We present the sequential algorithm \nlpagen and the first efficient parallel algorithm \pnlpagen for polynomial preferential attachment. 
Furthermore, we present the first I/O-efficient algorithm \emgenpa for general preferential attachment.

For a comparison with the state of the art, we engineer a sequential solution \dyngen that relies on the fully dynamic sampling data structure proposed by~\cite{DBLP:journals/mst/MatiasVN03}.
We find that \nlpagen performs better or similarly well as \dyngen; only for $\alpha > 1$ and adding many nodes, we find that the latter exhibits slightly better memory behavior due to the degenerate nature of the degree distribution in the limit.
In addition, our parallel algorithm \pnlpagen obtains a speed-up of $46$ for $\alpha \leq 1$ and $32$ for $\alpha > 1$.
We expect further improvements by using longer batches that may include multiple dependent samples.

Our experiments suggest that \emgenpa performs similarly well to the external memory state-of-the-art algorithm TFP-BA for the linear case.
Additionally, we investigate the performance of \emgenpa for different polynomial models.
For larger exponents, \emgenpa performs better due to more favorable output degree distributions enabling more cache-efficient degree sampling and incurring less I/Os from the underlying priority-queues.
While \emgenpa is feasible for virtually any set of realistic input parameters it would still be interesting to lift the restriction that the degrees are kept internally.

\clearpage

\bibliography{skewed-pa,dblp-downloaded}

\end{document}